\documentclass[twocolumn]{article}
\usepackage{amssymb}
\usepackage{amsthm,amsmath}

\title{\huge Structure functions and minimal path sets}
\author{Jean-Luc Marichal\footnote{Jean-Luc Marichal is with the Mathematics Research Unit, FSTC, University of Luxembourg, 6, rue Coudenhove-Kalergi, L-1359 Luxembourg, Luxembourg. Email: jean-luc.marichal[at]uni.lu}}
\date{Revised version, September 30, 2015}

\begin{document}
\maketitle

\theoremstyle{plain}
\newtheorem{theorem}{Theorem}
\newtheorem{lemma}[theorem]{Lemma}
\newtheorem{proposition}[theorem]{Proposition}
\newtheorem{corollary}[theorem]{Corollary}
\newtheorem{fact}[theorem]{Fact}
\newtheorem*{main}{Main Theorem}

\theoremstyle{definition}
\newtheorem{definition}[theorem]{Definition}
\newtheorem{example}{Example}
\newtheorem{algorithm}{Algorithm}

\theoremstyle{remark}
\newtheorem*{conjecture}{onjecture}
\newtheorem{remark}{Remark}
\newtheorem{claim}{Claim}

\newcommand{\N}{\mathbb{N}}
\newcommand{\R}{\mathbb{R}}
\newcommand{\Q}{\mathbb{Q}}
\newcommand{\Vspace}{\vspace{2ex}}
\newcommand{\bfx}{\mathbf{x}}
\newcommand{\bfy}{\mathbf{y}}
\newcommand{\bfz}{\mathbf{z}}
\newcommand{\bfh}{\mathbf{h}}
\newcommand{\bfe}{\mathbf{e}}
\newcommand{\bfs}{\mathbf{s}}
\newcommand{\bfS}{\overline{\mathbf{S}}}
\newcommand{\bfd}{\mathbf{d}}
\newcommand{\os}{\mathrm{os}}
\newcommand{\dd}{\,\mathrm{d}}

\begin{abstract}
In this short note we give and discuss a general multilinear expression of the structure function of an arbitrary semicoherent system in terms of its minimal path and cut sets. We also examine the link between the number of minimal path and cut sets consisting of 1 or 2 components and the concept of structure signature of the system.
\end{abstract}

\Vspace


\section*{Notation}

\begin{tabbing}
$[n]$  \hspace{8ex}\= set $\{1,\ldots,n\}$ \\
$C$ \> set of components of the system\\
$\phi(\bfx)$ \> structure function of the system\\
$\phi^D(\bfx)$ \> dual structure function \\
$h(\mathbf{p})$ \> reliability function of the system\\
$\mathbf{s}$ \> signature of the system \\
$s_k$ \> $k$-th coordinate of $\mathbf{s}$
\end{tabbing}

\section{Introduction}

Consider an $n$-component system $(C,\phi)$, where $C$ is the set $[n]=\{1,\ldots,n\}$ of its components and $\phi\colon\{0,1\}^n\to\{0,1\}$ is its structure function which expresses the state of the system in terms of the states of its components. We assume that the system is semicoherent, which means that the structure function is nondecreasing in each variable and satisfies the conditions $\phi(0,\ldots,0)=0$ and $\phi(1,\ldots,1)=1$.

Throughout we identify Boolean $n$-vectors $\bfx\in\{0,1\}^n$ and subsets $A\subseteq [n]$ in the usual way, that is, by setting $x_i=1$ if and only if $i\in A$. This identification enables us to use the same symbol to denote both a function $f\colon\{0,1\}^n\to\R$ and the corresponding set function $f\colon 2^{[n]}\to\R$ interchangeably. For instance, we write $\phi(0,\ldots,0)=\phi(\varnothing)$ and $\phi(1,\ldots,1)=\phi(C)$.

As a Boolean function, the structure function can always be written in the multilinear form
\begin{eqnarray}
\phi(\bfx) &=& \phi(x_1,\ldots,x_n)\nonumber\\
 &=& \sum_{A\subseteq C}\phi(A)\, \prod_{i\in A}x_i\,\prod_{i\in C\setminus A}(1-x_i).\label{eq:SelfDescForm}
\end{eqnarray}
Since the coefficients in this form are exactly the structure function values, we will refer to this form as the \emph{self-descriptive} form of the structure function. By expanding the second product in Eq.~(\ref{eq:SelfDescForm}) and then collecting terms, we obtain the \emph{simple} form of the structure function, namely
\begin{equation}\label{eq:SimpleForm}
\phi(\bfx) ~=~ \sum_{A\subseteq C}d(A)\, \prod_{i\in A}x_i{\,},
\end{equation}
where the link between the new coefficients $d(A)$ and the values $\phi(A)$, which can be obtained from the M\"obius inversion theorem, is given through the following linear conversion formulas (see, e.g., \cite[p.~31]{Ram90})
$$
\phi(A) ~=~ \sum_{B\subseteq A}d(B)
$$
and
$$
d(A) ~=~ \sum_{B\subseteq A}(-1)^{|A|-|B|}\, \phi(B)\, .
$$

Recall that a \emph{path set} of the system is a component subset $P\subseteq C$ such that $\phi(P)=1$. A path set $P$ of the system is said to be \emph{minimal} if $\phi(P')=0$ for every $P'\varsubsetneq P$. It is well known \cite[Ch.~2]{BarPro81} that if $P_1,\ldots, P_r$ denote the minimal path sets of the system, then
\begin{equation}\label{eq:sa7fd5ds}
\phi(\bfx) ~=~ \coprod_{j\in [r]}{\,}\prod_{i\in P_j}x_i ~=~ 1-\prod_{j\in [r]}\bigg(1-\prod_{i\in P_j}x_i\bigg),
\end{equation}
where $[r]=\{1,\ldots,r\}$ and $\amalg$ is the coproduct operation (i.e., the dual of the product operation) defined by $\amalg_i x_i=1-\Pi_i(1-x_i)$.

\begin{example}\label{ex:Bridge}
Consider the bridge structure as indicated in Figure~\ref{fig:bs}. This structure is characterized by four minimal path sets, namely $P_1=\{1,4\}$, $P_2=\{2,5\}$, $P_3=\{1,3,5\}$, and $P_4=\{2,3,4\}$. Equation~(\ref{eq:sa7fd5ds}) then shows that the structure function is given by
\begin{equation}\label{eq:BridgeAmalg}
\phi(x_1,\ldots,x_5) ~=~ x_1\, x_4\amalg x_2\, x_5\amalg x_1\, x_3\, x_5\amalg x_2\, x_3\, x_4{\,}.
\end{equation}
The simple form of the structure function can be easily computed by expanding the coproducts in (\ref{eq:BridgeAmalg}) and simplifying the resulting algebraic expression using $x_i^2=x_i$. We then obtain
\begin{eqnarray*}
\lefteqn{\phi(x_1,\ldots,x_5)}\\
&=& 1-(1-x_1x_4)(1-x_2x_5)(1-x_1x_3x_5)(1-x_2x_3x_4)\\
&=& x_1 x_4 + x_2 x_5 + x_1 x_3 x_5 + x_2 x_3 x_4 - x_1 x_2 x_3 x_4\\
&& \null - x_1 x_2 x_3 x_5 - x_1 x_2 x_4 x_5  - x_1 x_3 x_4 x_5 - x_2 x_3 x_4 x_5\\
&& \null + 2\, x_1 x_2 x_3 x_4 x_5\, ,
\end{eqnarray*}
which reveals the coefficients $d(A)$ of the simple form of the structure function.\qed
\end{example}

\setlength{\unitlength}{4ex}
\begin{figure}[htbp]\centering
\begin{picture}(11,4)
\put(3,0.5){\framebox(1,1){$2$}} \put(3,2.5){\framebox(1,1){$1$}} \put(5,1.5){\framebox(1,1){$3$}} \put(7,0.5){\framebox(1,1){$5$}}
\put(7,2.5){\framebox(1,1){$4$}}%
\put(0,2){\line(1,0){1.5}}\put(1.5,2){\line(2,-1){1.5}}\put(5.5,0){\line(-2,1){1.5}}\put(1.5,2){\line(2,1){1.5}}\put(5.5,4){\line(-2,-1){1.5}}%
\put(0,2){\circle*{0.15}}%
\put(9.5,2){\line(1,0){1.5}}\put(5.5,0){\line(2,1){1.5}}\put(9.5,2){\line(-2,-1){1.5}}\put(5.5,4){\line(2,-1){1.5}}\put(9.5,2){\line(-2,1){1.5}}%
\put(11,2){\circle*{0.15}}%
\put(5.5,0){\line(0,1){1.5}}\put(5.5,4){\line(0,-1){1.5}}
\end{picture}
\caption{Bridge structure} \label{fig:bs}
\end{figure}
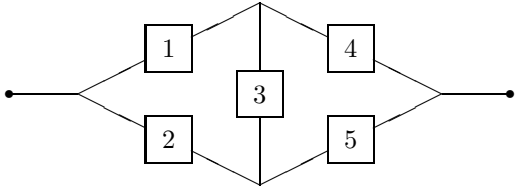

Example~\ref{ex:Bridge} illustrates the important fact that the simple form (\ref{eq:SimpleForm}) of any structure function can be expressed in terms of the minimal path sets of the system simply by expanding the coproduct in (\ref{eq:sa7fd5ds}) and then simplifying the resulting polynomial expression (using $x_i^2=x_i$) until it becomes multilinear. It seems, however, that such a general expression for the structure function is unknown in the literature.

In Section 2 of this note we yield an expression of the simple form of the structure function in terms of the minimal path sets. The derivation of this expression is inspired from the exact computation of the reliability function of the system by means of the inclusion-exclusion principle.
We also provide the dual version of this expression in terms of the minimal cut sets and discuss some interesting consequences of these expressions. In Section 3 we show that the number of minimal path and cut sets consisting of 1 or 2 components can be computed easily from the concept of structure signature of the system.

\section{Structure functions and minimal path and cut sets}

By extending formally the structure function to the hypercube $[0,1]^n$ by linear interpolation, we define its \emph{multilinear extension} (a concept introduced in game theory by Owen \cite{Owe72}) as the multilinear polynomial function $\widehat{\phi}\colon [0,1]^n\to [0,1]$ defined by
$$
\widehat{\phi}(\mathbf{x}) ~=~ \sum_{A\subseteq C}\phi(A)\prod_{i\in A}x_i\prod_{i\in C\setminus A}(1-x_i).
$$

Let $\phi^D\colon\{0,1\}^n\to\{0,1\}$ be the dual structure function defined as $\phi^D(\bfx)=1-\phi(\mathbf{1}-\bfx)$, where $\mathbf{1}$ stands for the $n$-vector $(1,\ldots,1)$, and let $d^D(A)$ be the coefficient of $\prod_{i\in A}x_i$ in the simple form of $\phi^D$. By using the dual structure function we can easily derive various useful forms of the structure function and its multilinear extension (see also Grabisch et al.~\cite{GraMarRou00}). Table~\ref{tab:StrFun} summarizes the best known of these forms (in addition to the minimal path set representation given in Eq.~(\ref{eq:sa7fd5ds})).

\begin{table}[htb]
$$
\begin{array}{|c|}
\hline
\begin{minipage}{0.4\textwidth}
\medskip
\begin{enumerate}
\item Self-descriptive form
$$
\sum_{A\subseteq C}\phi(A)\, \prod_{i\in A}x_i\,\prod_{i\in C\setminus A}(1-x_i)
$$
\item Dual self-descriptive form
$$
1-\sum_{A\subseteq C}\phi^D(A)\prod_{i\in C\setminus A}x_i\prod_{i\in C}(1-x_i)
$$
\item Simple form
$$
\sum_{A\subseteq C}d(A)\prod_{i\in A}x_i
$$
\item Dual simple form
$$
\sum_{A\subseteq C}d^D(A)\coprod_{i\in A}x_i
$$
\item Disjunctive normal form
$$
\coprod_{\textstyle{A\subseteq C\atop\phi(A)=1}}\prod_{i\in A} x_i
$$
\item Conjunctive normal form
$$
\prod_{\textstyle{A\subseteq C\atop\phi^D(A)=1}}\coprod_{i\in A} x_i
$$
\end{enumerate}
\null
\end{minipage}
\\
\hline
\end{array}
$$
\caption{Various forms of the structure function $\phi(\mathbf{x})$ and its multilinear extension $\widehat{\phi}(\mathbf{x})$} \label{tab:StrFun}
\end{table}

The concept of multilinear extension of the structure function has the following important interpretation in reliability theory. When the components are statistically independent, the function $\widehat{\phi}$ is nothing other than the \emph{reliability function} $h\colon [0,1]^n\to [0,1]$, which gives the reliability
$$
h(\mathbf{p}) ~=~ h(p_1,\ldots,p_n) ~=~ \sum_{A\subseteq C}\phi(A)\prod_{i\in A}p_i\prod_{i\in C\setminus A}(1-p_i)
$$
of the system in terms of the reliabilities $p_1,\ldots,p_n$ of the components (see, e.g., \cite[Ch.~2]{BarPro81}).

The exact computation of the system reliability $h(\mathbf{p})$ in terms of the minimal path sets $P_1,\ldots,P_r$ is usually done by means of the inclusion-exclusion method (see, e.g., \cite[Sect.~6.2]{Bar98} and \cite[Ch.~2]{BarPro81}). In this section we recall this method and show how we can adapt it to derive a concise expression of the simple form of the structure function in terms of the minimal path sets of the system.

For every $j\in [r]=\{1,\ldots,r\}$, let $E_j$ be the event that all components in the minimal path set $P_j$ work. Then, using the inclusion-exclusion formula for probabilities, we obtain
\begin{eqnarray}
h(\mathbf{p}) &=& \Pr\bigg(\bigcup_{j\in [r]}E_j\bigg)\nonumber\\
&=& \sum_{\varnothing\neq B\subseteq [r]} (-1)^{|B|-1}{\,}\Pr\bigg(\bigcap_{j\in B}E_j\bigg){\,}.\label{eq:hpppIEm}
\end{eqnarray}
Let $D_i$ denote the event that component $i$ works. Then, we have $p_i=\Pr(D_i)$ and, using the independence assumption, we have
$$
\Pr(E_j) ~=~ \Pr\bigg(\bigcap_{i\in P_j}D_i\bigg) ~=~ \prod_{i\in P_j} p_i
$$
and, more generally,
\begin{eqnarray}
\lefteqn{\Pr\bigg(\bigcap_{j\in B}E_j\bigg) ~=~ \Pr\bigg(\bigcap_{j\in B}{\,}\bigcap_{i\in P_j}D_i\bigg)}\nonumber\\
&=& \Pr\bigg(\bigcap_{i\in\bigcup_{j\in B}P_j}D_i\bigg) ~=~ \prod_{i\in\bigcup_{j\in B}P_j} p_i{\,}.\label{eq:PrCapEj}
\end{eqnarray}
Substituting (\ref{eq:PrCapEj}) in (\ref{eq:hpppIEm}), we obtain the following multilinear expression of $h(\mathbf{p})$ in terms of the minimal path sets of the system
\begin{equation}\label{eq:ReFctMinPSets}
h(\mathbf{p}) ~=~ \sum_{\varnothing\neq B\subseteq [r]} (-1)^{|B|-1}{\,}\prod_{i\in\bigcup_{j\in B}P_j} p_i{\,},
\end{equation}
or equivalently,
\begin{eqnarray*}
h(\mathbf{p}) &=& \sum_{j\in [r]}{\,}\prod_{i\in P_j}p_i {~} - \sum_{\{j,k\}\subseteq [r]}{\,}\prod_{i\in P_j\cup P_k}p_i\\
&& \null + \sum_{\{j,k,l\}\subseteq [r]}{\,}\prod_{i\in P_j\cup P_k\cup P_l}p_i {~} -\mbox{}\cdots
\end{eqnarray*}

We now show that a similar formula can be obtained for the structure function without an appeal to the independence assumption on the system components. Actually, our result and its proof are purely combinatorial and does not need any stochastic setting.

\begin{theorem}\label{thm:mainPrimal}
If $P_1,\ldots, P_r$ denote the minimal path sets of the system, then
\begin{equation}\label{eq:sa7fd5ds2}
\phi(\bfx) ~=~ \sum_{\varnothing\neq B\subseteq [r]}(-1)^{|B|-1}\prod_{i\in\bigcup_{j\in B}P_j}x_i{\,},
\end{equation}
or equivalently,
\begin{eqnarray*}
\phi(\bfx) &=& \sum_{j\in [r]}{\,}\prod_{i\in P_j}x_i {~} - \sum_{\{j,k\}\subseteq [r]}{\,}\prod_{i\in P_j\cup P_k}x_i\\
&& \null + \sum_{\{j,k,l\}\subseteq [r]}{\,}\prod_{i\in P_j\cup P_k\cup P_l}x_i {~} -\mbox{}\cdots
\end{eqnarray*}
\end{theorem}

Clearly, Eq.~(\ref{eq:sa7fd5ds2}) still holds on $[0,1]^n$ if we replace the structure function with its multilinear extension. In particular, when the components are statistically independent, we see that (\ref{eq:ReFctMinPSets}) immediately follows from (\ref{eq:sa7fd5ds2}).

\begin{example}\label{ex:sad75f55}
Consider a $4$-component system defined by the three minimal path sets $P_1=\{1,2\}$, $P_2=\{2,3\}$, and $P_3=\{3,4\}$.
\begin{table}[htbp]
$$
\begin{array}{|c|c|c|}
\hline
B & (-1)^{|B|-1} & \bigcup_{j\in B}P_j^{\mathstrut}\\
\hline
\{1\} & 1 & \{1,2\}\\
\{2\} & 1 & \{2,3\}\\
\{3\} & 1 & \{3,4\}\\
\{1,2\} & -1 & \{1,2,3\}\\
\{1,3\} & -1 & \{1,2,3,4\}\\
\{2,3\} & -1 & \{2,3,4\}\\
\{1,2,3\} & 1 & \{1,2,3,4\}\\
\hline
\end{array}
$$
\caption{Example~\ref{ex:sad75f55}}
\label{tab:113}
\end{table}
The constituting elements of the sum in Eq.~(\ref{eq:sa7fd5ds2}) are gathered in Table~\ref{tab:113}. Summing up the monomials defined by the subsets given in the third column, each multiplied by the corresponding number ($+1$ or $-1$) from the second column, by (\ref{eq:sa7fd5ds2}) we obtain
$$
\phi(x_1,x_2,x_3,x_4) ~=~ x_1x_2+x_2x_3+x_3x_4-x_1x_2x_3-x_2x_3x_4{\,},
$$
which is the simple form of the structure function.\qed
\end{example}


Interestingly, Theorem~\ref{thm:mainPrimal} enables us to identify the minimal path sets of the system from the simple form of the structure function by quick inspection. We state this result in the following immediate corollary.

\begin{corollary}\label{cor:dsa67}
The minimal path sets $P_1,\ldots, P_r$ are exactly the minimal elements (with respect to inclusion) of the family of subsets defined by the monomials (or equivalently, the monomials with coefficient $+1$) in the simple form of the structure function.
\end{corollary}

Corollary~\ref{cor:dsa67} enables us to reconstruct the minimal path set representation (\ref{eq:sa7fd5ds}) of the structure function from its simple form. Considering for instance the simple form of the bridge structure function as described in Example~\ref{ex:Bridge}, by Corollary~\ref{cor:dsa67} we see that the corresponding minimal path sets are $P_1=\{1,4\}$, $P_2=\{2,5\}$, $P_3=\{1,3,5\}$, and $P_4=\{2,3,4\}$. We then immediately retrieve Eq.~(\ref{eq:BridgeAmalg}).

Theorem~\ref{thm:mainPrimal} has the following additional consequence. Recall that a \emph{formation} of a subset $A$ of $C$ is a collection of minimal path sets whose union is $A$. A formation of $A$ is said to be \emph{odd} (resp.\ \emph{even}) if it is the union of an odd (resp.\ even) number of minimal path sets. Note that a particular formation can be both odd and even. By equating the corresponding terms in (\ref{eq:SimpleForm}) and (\ref{eq:sa7fd5ds2}), we obtain the following identity
$$
d(A) ~=~ \sum_{\textstyle{\varnothing\neq B\subseteq [r]\atop \bigcup_{j\in B}P_j=A}}(-1)^{|B|-1}.
$$
From this identity, we immediately retrieve the important fact (see, e.g., Barlow and Iyer~\cite{BarIye88}) that the coefficient $d(A)$ is exactly the number of odd formations of $A$ minus the number of even formations of $A$.

A dual argument enables us to yield an expression of the simple form of the structure function in terms of the minimal cut sets of the system. Recall that a subset $K$ of $C$ is a \emph{cut set} of the system if $\phi(C\setminus K)=0$. It is \emph{minimal} if $\phi(C\setminus K')=1$ for every $K'\varsubsetneq K$. If $K_1,\ldots,K_s$ denote the minimal cut sets of the system, then
$$
\phi(\bfx) ~=~ \prod_{j\in [s]}\coprod_{i\in K_j}x_i ~=~ \prod_{j\in [s]}\bigg(1-\prod_{i\in K_j}(1-x_i)\bigg).
$$

Starting from the well-known fact that the minimal cut sets of the system are the minimal path sets of the dual, and vice versa, from Theorem~\ref{thm:mainPrimal} and Corollary~\ref{cor:dsa67} we immediately derive the following dual versions.

\begin{theorem}\label{thm:mainDual}
If $K_1,\ldots, K_s$ denote the minimal cut sets of the system, then
\begin{equation}\label{eq:sa7fd5ds2d}
\phi^D(\bfx) ~=~ \sum_{\varnothing\neq B\subseteq [s]}(-1)^{|B|-1}\prod_{i\in\bigcup_{j\in B}K_j}x_i{\,}.
\end{equation}
\end{theorem}

\begin{corollary}\label{cor:dsa671}
The minimal cut sets $K_1,\ldots, K_s$ are exactly the minimal elements (with respect to inclusion) of the family of subsets defined by the monomials (or equivalently, the monomials with coefficient $+1$) in the simple form of the dual structure function.
\end{corollary}

Locks~\cite{Loc78} described a method for generating all minimal cut sets from the set of minimal path sets (and vice versa) using Boolean algebra. Interestingly, an alternative method simply consists in applying Corollary~\ref{cor:dsa671} to the dual structure function
$$
\phi^D(\bfx) ~=~ 1-\phi(\mathbf{1}-\bfx) ~=~ \prod_{j\in [r]}\coprod_{i\in P_j}x_i{\,}.
$$
Consider for instance the structure function defined in Example~\ref{ex:sad75f55}. The dual structure function is given by
\begin{eqnarray*}
\phi^D(\bfx) &=& (x_1\amalg x_2)(x_2\amalg x_3)(x_3\amalg x_4)\\
&=& x_1x_3+x_2x_3+x_2x_4-x_1x_2x_3-x_2x_3x_4{\,}.
\end{eqnarray*}
Corollary~\ref{cor:dsa671} then immediately yields the minimal cut sets of the system, namely $K_1=\{1,3\}$, $K_2=\{2,3\}$, and $K_3=\{2,4\}$.

\begin{remark}
It is noteworthy that from (\ref{eq:hpppIEm}) and (\ref{eq:PrCapEj}) we can immediately derive a representation of the system reliability function in terms of the reliability functions of the series systems defined from the unions of minimal path sets. More precisely, for every $t>0$ we have
$$
\Pr(T>t) ~=~ \sum_{\varnothing\neq B\subseteq [r]} (-1)^{|B|-1}{\,}\Pr\bigg(\min_{i\in\bigcup_{j\in B}P_j}T_i>t\bigg){\,},
$$
where $T$ and $T_i$ denote the lifetime of the system and the lifetime of component $i$, respectively. This representation, which holds regardless of the distribution of the component lifetimes, has been obtained for instance in \cite[Eq.~(3.1)]{BloLiSav03}, \cite[Eq.~(2.2)]{NavAguSorSua14}, and \cite[Eq.~(3.4)]{NavRuiSan07}. The corresponding dual version can be easily derived by considering parallel systems and minimal cut sets; see for instance \cite[Eq.~(3.5)]{NavRuiSan07}. For every $t>0$ we have
$$
\Pr(T\leqslant t) ~=~ \sum_{\varnothing\neq B\subseteq [s]} (-1)^{|B|-1}{\,}\Pr\bigg(\max_{i\in\bigcup_{j\in B}K_j}T_i\leqslant t\bigg){\,}.
$$
\end{remark}

\section{Minimal path and cut sets of small sizes}

By identifying the variables $x_1,\ldots,x_n$ in the multilinear extension $\widehat{\phi}(\bfx)$ of the structure function, we define its diagonal section $\widehat{\phi}(x,\ldots,x)$, which will be simply denoted by $\widehat{\phi}(x)$. From the simple form (\ref{eq:SimpleForm}) of the structure function, we immediately obtain the polynomial function
$$
\widehat{\phi}(x) ~=~ \sum_{k=1}^n d_k\, x^{k}{\,},
$$
where
$$
d_k ~=~ \sum_{\textstyle{A\subseteq C\atop |A|=k}}d(A){\,}.
$$
For instance, considering the bridge structure function defined in Example~\ref{ex:Bridge}, we obtain $\widehat{\phi}(x) = 2x^2+2x^3-5x^4+2x^5$.

By definition, the diagonal section of the multilinear extension of the structure function is also the one-variable \emph{reliability function} $h\colon [0,1]\to [0,1]$ which gives the system reliability $h(p)=h(p,\ldots,p)$ of the system whenever the components are statistically independent and have the same reliability $p$.

Using Theorem~\ref{thm:mainPrimal}, we can easily express the function $\widehat{\phi}(x)$ in terms of the minimal path sets. We simply have
$$
\widehat{\phi}(x) ~=~ \sum_{\varnothing\neq B\subseteq [r]}(-1)^{|B|-1}{\,}x^{|\bigcup_{j\in B}P_j|}
$$
and the coefficient $d_k$ of $x^k$ in $\widehat{\phi}(x)$ is then given by
\begin{equation}\label{eq:ai777}
d_k ~=~ \sum_{\textstyle{B\subseteq [r]\atop |\bigcup_{j\in B}P_j|=k}}(-1)^{|B|-1}{\,}.
\end{equation}
Dually, the coefficient $d^D_k$ of $x^k$ in $\widehat{\phi}^D(x)$ is given by
\begin{equation}\label{eq:bi777}
d_k^D ~=~ \sum_{\textstyle{B\subseteq [s]\atop |\bigcup_{j\in B}K_j|=k}}(-1)^{|B|-1}{\,}.
\end{equation}

For every $k\in [n]$, let $\alpha_k$ (resp.\ $\beta_k$) denote the number of minimal path (resp.\ cut) sets of size $k$ of the system. The knowledge of these numbers for small $k$ may be relevant when analyzing the reliability of the system. For instance, if the system has no minimal cut set of size $1$, it may be informative to count the number $\beta_2$ of minimal cut sets of size $2$ and so forth.

The following proposition shows that $\alpha_1$ and $\alpha_2$ (resp.\ $\beta_1$ and $\beta_2$) can be computed directly from the coefficients $d_1$ and $d_2$ (resp.\ $d_1^D$ and $d_2^D$), and vice versa.

\begin{proposition}\label{prop:Conv}
We have $\alpha_1 = d_1$, $\beta_1 = d^D_1$, $\alpha_2 = {d_1\choose 2}+d_2$, and $\beta_2 = {d^D_1\choose 2}+d^D_2$.
\end{proposition}

The following example shows that, in general, for $k\geqslant 3$ neither $\alpha_k$ nor $\beta_k$ can be determined only from the coefficients of $\widehat{\phi}(x)$ and $\widehat{\phi}^D(x)$.

\begin{example}
Consider the structure functions
\begin{eqnarray}
\phi_1(\bfx) &=& x_1 x_2 \amalg x_3 x_4\nonumber\\
&=& x_1x_2+x_3x_4-x_1x_2x_3x_4\label{eq:ex-1}
\end{eqnarray}
and
\begin{eqnarray*}
\phi_2(\bfx) &=& x_1 x_2 \amalg x_1 x_3 \amalg x_2 x_3 x_4 \\
&=& x_1x_2+x_1x_3+x_2x_3x_4-x_1x_2x_3-x_1x_2x_3x_4{\,}.
\end{eqnarray*}
We have $\widehat{\phi}_1(x)=\widehat{\phi}_2(x)=2x^2-x^4$ and hence $\widehat{\phi}^D_1(x)=1-\widehat{\phi}_1(1-x)=1-\widehat{\phi}_2(1-x)=\widehat{\phi}^D_2(x)$. However, we clearly have $\alpha_3=0$ for function $\phi_1$ and $\alpha_3=1$ for function $\phi_2$.
\end{example}

\begin{remark}
Contrary to the number $\alpha_k$, the number of path sets (not necessarily minimal) of size $k$ can always be determined from the coefficients of $\widehat{\phi}(x)$. Indeed, this number is given by (see, e.g., \cite[Prop.~1]{Marb})
$$
\sum_{\textstyle{A\subseteq C\atop |A|=k}}\phi(A) ~=~ \sum_{j=0}^k {n-j\choose k-j}{\,}d_j{\,}.
$$
Dually, the number of cut sets of size $k$ is given by
$$
\sum_{\textstyle{A\subseteq C\atop |A|=k}}\phi^D(A) ~=~ \sum_{j=0}^k {n-j\choose k-j}{\,}d_j^D{\,}.
$$
\end{remark}

We end this section by giving expressions for the numbers $\alpha_1$, $\alpha_2$, $\beta_1$, and $\beta_2$ in terms of the structure signature of the system.

Recall that the \emph{structure signature} of the system is the $n$-vector $\bfs=(s_1,\ldots,s_n)$ whose $k$-th coordinate is defined as
\begin{equation}\label{eq:asad678}
s_k ~=~ \sum_{\textstyle{A\subseteq C\atop |A|=n-k+1}}\frac{1}{{n\choose |A|}}\,\phi(A)-\sum_{\textstyle{A\subseteq C\atop |A|=n-k}}\frac{1}{{n\choose |A|}}\,\phi(A)\, ,
\end{equation}
or equivalently,
\begin{equation}\label{eq:Ss2}
s_k ~=~ \sum_{j=1}^{n-k+1}\frac{{n-k\choose j-1}}{{n\choose j}}{\,} d_j{\,}.
\end{equation}

This concept was introduced in 1985 by Samaniego~\cite{Sam85} for systems whose components have continuous and i.i.d.\ lifetimes. He originally defined $s_k$ as the probability that the $k$-th component failure causes the system to fail (hence the property $\sum_{k=1}^n s_k=1$). More recently, Boland~\cite{Bol01} showed that this probability can be explicitly given by (\ref{eq:asad678}). The expression given in (\ref{eq:Ss2}) was derived later in \cite[Cor.~12]{Mar14} and \cite[Prop.~3]{Marb} (see also \cite{MarMatb} for a preliminary work). Thus defined, the structure signature depends only on the structure function and can actually be considered for any system, without any assumption on the distribution of the component lifetimes.

Combining this concept with Proposition~\ref{prop:Conv} shows that $\alpha_1$ and $\alpha_2$ (resp.\ $\beta_1$ and $\beta_2$) can be computed directly from $s_n$ and $s_{n-1}$ (resp.\ $s_1$ and $s_2$), and vice versa. The conversion formulas are given in the following proposition.

\begin{proposition}\label{prop:Conv2}
We have
\begin{eqnarray}
\alpha_1 &=& d_1 ~=~ n s_n{\,},\label{eq:tz1}\\
\beta_1 &=& d^D_1 ~=~ n s_1{\,},\label{eq:tz2}\\
\alpha_2 &=& \textstyle{{d_1\choose 2}+d_2} ~=~ \textstyle{{n s_n\choose 2}+{n\choose 2} (s_{n-1}-s_n)}{\,},\label{eq:tz3}\\
\beta_2 &=& \textstyle{{d^D_1\choose 2}+d^D_2} ~=~ \textstyle{{n s_1\choose 2}+{n\choose 2} (s_2-s_1)}{\,}.\label{eq:tz4}
\end{eqnarray}
\end{proposition}

For instance, consider again the structure function $\phi_1$ given in (\ref{eq:ex-1}). We have $\widehat{\phi}_1(x)=2x^2-x^4$, $\widehat{\phi}_1^D(x)=4x^2-4x^3+x^4$, and therefore $d_1=d_1^D=0$, $d_2=2$, and $d_2^D=4$. Using Proposition~\ref{prop:Conv2}, we finally obtain $\alpha_1=\beta_1=0$, $\alpha_2=2$, $\beta_2=4$, and $\bfs=(0,\frac{2}{3},\frac{1}{3},0)$.

\begin{example}
Consider an $n$-component system having $\beta_2$ minimal cut sets of size $2$ and no cut set of size $1$. By (\ref{eq:tz2}) and (\ref{eq:tz4}) we necessarily have $s_1=0$ and $s_2=\beta_2/{n\choose 2}$. This result was expected since $s_k$ is the probability that, assuming that the component lifetimes are continuous and i.i.d.\ (and hence exchangeable), the system fails exactly at the $k$-th component failure. Thus, $s_1$ is clearly zero and $s_2$ is the ratio of the number $\beta_2$ of minimal cut sets of size $2$ (favorable cases) over the number ${n\choose 2}$ of pairs of components (possible cases).\qed
\end{example}

\begin{remark}\label{rem:45krwe}
It is noteworthy that, combining (\ref{eq:ai777}) with (\ref{eq:Ss2}), we obtain a simple way to compute the structure signature of the system directly from the minimal path sets. Dually, combining (\ref{eq:bi777}) with the immediate formula (see also \cite[Sect.~3.5]{Marb})
$$
s_k ~=~ \sum_{j=1}^k\frac{{k-1\choose j-1}}{{n\choose j}}{\,} d_j^D{\,}
$$
shows how we can compute the structure signature directly from the minimal cut sets.
\end{remark}

\begin{remark}
One can easily show that, when the components have exchangeable lifetimes, the system reliability function can be expressed as
$$
\Pr(T>t) ~=~ \sum_{k=1}^n d_k{\,}\Pr\big(\min\{T_1,\ldots,T_k\}>t\big){\,}.
$$
This result shows that the $n$-vector $\bfd=(d_1,\ldots,d_n)$ can be interpreted as a signature vector, called ``minimal signature'' in \cite[Def.~4.1]{NavRuiSan07}. The structure signature $\bfs$ can then be computed from this minimal signature by using (\ref{eq:Ss2}) and vice versa (see, e.g., \cite{Marb}). Dually, one can show that
$$
\Pr(T\leqslant t) ~=~ \sum_{k=1}^n d_k^D{\,}\Pr\big(\max\{T_1,\ldots,T_k\}\leqslant t\big){\,},
$$
and this formula then shows that the coefficients $d_1^D,\ldots,d_n^D$ can be used to define the ``maximal signature'' \cite[Def.~4.2]{NavRuiSan07}, which also determines the structure signature $\bfs$ (as indicated in Remark~\ref{rem:45krwe}) and the $n$-vector $\bfd$ (and vice versa).
\end{remark}

\section*{Acknowledgments}

This research is supported by the internal research project F1R-MTH-PUL-15MRO3 of the University of Luxembourg.

\appendix\section*{Appendix}

\begin{proof}[Proof of Theorem~\ref{thm:mainPrimal}]
The proof relies on the classical polynomial inclusion-exclusion identity
$$
1-\prod_{j\in [r]}(1-z_j)=\sum_{\varnothing\neq B\subseteq [r]}(-1)^{|B|-1}\prod_{j\in B}z_j{\,},
$$
which holds for all $z_1,\ldots,z_n\in\R$. Setting $z_j=\prod_{i\in P_j}x_i$ in the latter identity and then combining the resulting formula with the right-hand expression in (\ref{eq:sa7fd5ds}), we immediately obtain
$$
\phi(\bfx) ~=~ \sum_{\varnothing\neq B\subseteq [r]}(-1)^{|B|-1}\prod_{j\in B}{\,}\prod_{i\in P_j}x_i.
$$
Formula (\ref{eq:sa7fd5ds2}) then follows by simplifying the latter expression using $x_i^2=x_i$.
\end{proof}

\begin{proof}[Proof of Proposition~\ref{prop:Conv}]
On the one hand, setting $k=1$ in (\ref{eq:ai777}) and (\ref{eq:bi777}) shows that $d_1=\alpha_1$ and $d_1^D=\beta_1$. On the other hand, setting $k=2$ in (\ref{eq:ai777}), we obtain
$$
d_2 ~=~ \big|\big\{i\in [r]:|P_i|=2\big\}\big| - \big|\big\{\{i,j\}\subseteq [r]:|P_i\cup P_j|=2\big\}\big|{\,},
$$
that is $d_2=\alpha_2-{\alpha_1\choose 2}$. Dually, we obtain $d^D_2=\beta_2-{\beta_1\choose 2}$.
\end{proof}

\begin{proof}[Proof of Proposition~\ref{prop:Conv2}]
From Eq.~(\ref{eq:Ss2}) we immediately derive the equations $d_1 = ns_n$ and $d_2 = {n\choose 2}(s_{n-1}-s_n)$. Now, if $\bfs^D=(s_1^D,\ldots,s_n^D)$ denotes the structure signature associated with the dual structure function $\phi^D$, then we have $s_k^D=s_{n+1-k}$ for $k=1,\ldots,n$. Combining this observation with the previous two equations, we obtain immediately $d_1^D = ns_1$ and $d_2^D = {n\choose 2}(s_2-s_1)$. We then conclude by Proposition~\ref{prop:Conv}.
\end{proof}

\noindent\textbf{Jean-Luc Marichal} is currently an associate professor in the Mathematics Research Unit at the University of Luxembourg. He received his Ph.D. in Mathematics at the University of Li{\`e}ge (Belgium) in 1998. His research area mainly includes aggregation function theory, functional equations, non-additive measures and integrals, conjoint measurement theory, cooperative game theory, and system reliability theory.

\end{document}